\documentclass[sigconf]{acmart}  

\usepackage{booktabs}

\setcopyright{rightsretained}  
\acmDOI{}  
\acmISBN{}  
\acmConference{}{}{}  
\acmYear{}  
\copyrightyear{2018}  
\acmPrice{}
\usepackage{xspace}
\newcommand{\eg}{e.g.,\xspace}
\newcommand{\ie}{i.e.,\xspace}

\usepackage{multirow}

\usepackage{graphicx}
\usepackage[font=small,skip=0pt]{caption}
\usepackage{epstopdf}
\usepackage{amsfonts}
\usepackage{amsthm}

\usepackage{algorithm}
\usepackage[noend]{algpseudocode}
\algnewcommand\algorithmicswitch{\textbf{switch}}
\algnewcommand\algorithmiccase{\textbf{case}}
\algdef{SE}[SWITCH]{Switch}{EndSwitch}[1]{\algorithmicswitch\ #1\ \algorithmicdo}{\algorithmicend\ \algorithmicswitch}%
\algdef{SE}[CASE]{Case}{EndCase}[1]{\algorithmiccase\ #1}{\algorithmicend\ \algorithmiccase}%
\algtext*{EndSwitch}%
\algtext*{EndCase}%
\def\NoNumber#1{{\def\alglinenumber##1{} #1}\addtocounter{ALG@line}{-1}}

\makeatletter
\newcommand{\StatexIndent}[1][3]{%
  \setlength\@tempdima{\algorithmicindent}%
  \Statex\hskip\dimexpr#1\@tempdima\relax}
\makeatother

\newcommand{\spec}{\ensuremath{\#}}
\newcommand{\atlr}{RB$\pm$ATL\xspace}

\newcommand\nat{\mathbb{N}}

\newcommand\integer{\mathbb{Z}}

\newcommand{\llangle}{\langle\!\langle}
\newcommand{\rrangle}{\rangle\!\rangle}

\newcommand\idle{idle\xspace}

\newcommand\AO[2]{\llangle #1 \rrangle \!\bigcirc\! #2}
\newcommand\AG[2]{\llangle #1 \rrangle \Box #2}
\newcommand\AU[3]{\llangle #1 \rrangle #2 \,{\mathcal U}\, #3}
\newcommand{\Agt}{Agt}
\newcommand{\Res}{Res}
\newcommand{\Act}{Act}

\newcommand{\untilstrategy}{until}
\newcommand{\hd}[1]{\ensuremath{\mathit{hd}(#1)}}
\newcommand{\tl}[1]{\ensuremath{\mathit{tl}(#1)}}
\mathchardef\mhyphen="2D

\newcommand{\nextstrategy}{next}

\newcommand{\onepath}[1][]{\ensuremath{\lambda\ifthenelse{\equal{#1}{}}{}{[#1]}}} 
\newcommand{\States}{\ensuremath{S}\xspace}
\newcommand{\coop}[2][]{\langle\!\langle{#2}\rangle\!\rangle_{_{\!\mathit{#1}}}}
\newcommand{\Props}{\Pi}
\newcommand{\Next}{\ensuremath{\!\bigcirc\!}}
\newcommand{\NUntil}{\ensuremath{\,\mathcal{U}}}
\newcommand{\NRel}{\ensuremath{\mathcal{R}}}

\newcommand{\rhooutcome}[3]{\ensuremath{\mathit{out}(#1,#2,#3)}}

\newcommand{\model}{\ensuremath{M}\xspace}
\newcommand{\enment}{\ensuremath{\eta}\xspace}

\newcommand{\production}{\mathsf{prod}}
\newcommand{\consumption}{\mathsf{cons}}
\newcommand{\Enments}{\ensuremath{\mathsf{En}}}

\newcommand{\modelsR}{\ensuremath{\models_R}\xspace}

\newcommand{\coopdown}[2][]{\ensuremath{\coop{#2}{}_{#1}^{\downarrow}}}

\renewcommand{\modelsR}{\models}

\newcommand\AR[3]{\llangle #1 \rrangle #2 \,{\mathcal R}\, #3}
\newcommand{\releasestrategy}{release}
\newcommand{\atlrd}{$\text{RB}\pm\text{ATL}^{\#}$\xspace}
\newcommand{\atlrir}{$\text{RB}\pm\text{ATL}_{iR}$\xspace}
\newcommand{\atlrdir}{$\text{RB}\pm\text{ATL}_{iR}^{\#}$\xspace}
\newcommand{\rald}{$\text{RAL}^{\#}$\xspace}

\begin{document}

\title{Resource Logics with a Diminishing Resource}  

\author{Natasha Alechina}
\orcid{0000-0003-3306-9891}
\affiliation{%
  \institution{University of Nottingham}
  \streetaddress{Jubilee Campus}
  \city{Nottingham} 
  \country{UK}
}
\email{nza@cs.nott.ac.uk}

\author{Brian Logan}
\orcid{0000-0003-0648-7107}
\affiliation{%
  \institution{University of Nottingham}
  \streetaddress{Jubilee Campus}
  \city{Nottingham}
  \country{UK}
}
\email{bsl@cs.nott.ac.uk}

\begin{abstract}  
Model-checking resource logics with production and consumption of resources is a computationally hard and often undecidable problem. We introduce a simple and realistic assumption that there is at least one \emph{diminishing resource}, that is, a resource that cannot be produced and every action has a non-zero cost on this resource. An example of such resource is time. We show that, with this assumption, problems that are undecidable even for the underlying Alternating Time Temporal Logic, such as model-checking under imperfect information and perfect recall, become decidable for resource logics with a diminishing resource. 
\end{abstract}

\maketitle

\section{Introduction}

There has been a considerable amount of work on multi-agent temporal logics 
interpreted over structures where agents' actions consume resources, or
both produce and consume resources. Examples include an extension of Coalition
Logic where actions consume resources and coalitional modalities are annotated
with resource bounds (`agents in coalition $A$ have a strategy of cost at most
$b$ to achieve $\phi$') (RBCL) \cite{Alechina//:09b,Alechina//:10d}, a similar extension
for Alternating Time Temporal Logic ATL (RB-ATL) \cite{Alechina//:10a}, 
extensions of Computation Tree Logic and Alternating Time Temporal Logic with 
both consumption and production of resources (RTL, RAL) 
\cite{BullingFarwer09rtl-clima-post,Bulling/Farwer:10a}, a variant of resource
bounded ATL where all resources are convertible to money and the amount of money
is bounded (PRB-ATL) \cite{DellaMonica//:11a,DellaMonica//:13a},
an extension of PRB-ATL to $\mu$-calculus \cite{DellaMonica/Lenzi:12a},
a version of ATL with more general numerical constraints (QATL${}^*$) 
\cite{Bulling/Goranko:13a}, a version of RB-ATL where unbounded
production of resources is allowed (\atlr) \cite{Alechina//:16d,Alechina//:16c}.
The model-checking problem for such resource logics is decidable, though often not comptationally tractable, when resources are only consumed or where the amount of resources is somehow bounded.
\cite{Bulling/Farwer:10a}. For RAL with unbounded production of resources,
the  model-checking problem is undecidable, and this holds even for several of its fragments \cite{Bulling/Farwer:10a}, although recently a fragment of RAL without the boundedness assumption has been found where the model-checking problem 
is decidable \cite{Alechina//:17b}. 
A slightly different semantics compared to RAL, but 
allowing unbounded production of resources, also
results in a decidable model-checking problem for resource extensions
of ATL such as \atlr \cite{Alechina//:16d}; the complexity of the model-checking problem for \atlr has been
shown to be 2EXPTIME-complete in \cite{Alechina//:16c}.

There exists also a large body of related work on reachability and non-termination problems in energy games and games on
vector addition systems with state 
\cite{Brazdil&Jancar&Kucera10,Jurdzinski&Lazic&Schmitz15}.
In fact, complexity and decidability results for resource logics in \cite{Alechina//:16c} build on 
the results for single-sided vector addition systems with states \cite{Courtois&Schmitz14,Abdullaetal13}. 

As far as we are aware, there is no work on model-checking resource logics under imperfect information.
For ATL (without resources) under imperfect information and with perfect recall uniform strategies
the problem is undecidable for three or more agents \cite{Dima/Tiplea:11a}. It is however decidable in the 
case of bounded strategies \cite{Vester:13a}.  For two player energy games with imperfect information and a fixed initial credit the existence of a winning strategy is also decidable \cite{Degorre//:10a}. 

In this paper we consider a special kind of models for resource logics satisfying a restriction
that one of the resources
is always consumed by each action. It is a very natural setting which occurs in many verification
problems for resource logics. The first obvious example of such a resource is time. Time is always
`consumed' by each action, and no agent in the system can turn back the clock and `produce' time. 
When a verification problem has time as one of the explicit resource parameters, the restriction certainly
applies. 
Other examples include 
systems where agents have a non-rechargeable battery 
and where all actions consume energy, e.g,. nodes in a wireless sensor network; and systems where agents have a store of propellant that cannot be replenished during the course of a  mission and all actions of interest involve manoeuvring, e.g., a constellation of satellites.
We call this special resource that is consumed by all actions a
\emph{diminishing resource}.

From the technical point of view, the restriction to systems with a diminishing resource has the advantage
that all strategies become bounded, even if for other resource types unbounded production is allowed. 
In the case of \atlr with a diminishing resource where the model-checking problem is already known to be
decidable and 2EXPTIME-complete, we can produce simpler model-checking algorithms and a lower complexity
bound (PSPACE if resource bounds are written in unary).
In the case of \atlr with a 
diminishing resource under imperfect information, the result of \cite{Vester:13a} does not apply 
immediatelly because the bound is not fixed in advance, but the logic is indeed decidable and we get a new 
set of model-checking algorithms and a complexity bound.
Finally, the decidability of RAL with a diminishing resource follows from the result on the decidability of RAL
on bounded models \cite{Bulling/Farwer:10a}, but the model-checking algorithms and the PSPACE upper bound (for resource endowments written in unary) are specific to RAL with diminishing resource and are new.

The rest of the paper is organised as follows. In Section \ref{sec:atlr-spec}, we introduce \atlrd
with a diminishing resource, motivate changes to its syntax (we use the Release operator instead of `Always'
or `Globally', and do not allow infinite resource bounds), give a
model-checking algorithm and analyse its complexity. In Section
\ref{sec:atlrdir} we introduce \atlrdir,
which is \atlrd under imperfect information and perfect recall, and give
a model-checking algorithm for it and analyse its complexity. Finally in
Section \ref{sec:rald} we define RAL with diminishing resource, give a model-checking algorithm for it and show that the complexity is the same as for \atlrd.

\section{\atlrd}
\label{sec:atlr-spec}

The logic \atlr was introduced in \cite{Alechina//:14c}, and its model-checking complexity studied in more detail in \cite{Alechina//:16d} and \cite{Alechina//:16c}. Here we consider a variant of this logic without the \idle action which is interpreted on finite paths. It contains a Release operator instead of Globally and does not allow infinite values in resource bounds. We use Release because it is not definable in ATL in terms
of Next, Until and Globally \cite{Laroussinie//:08a} while Globally is definable in terms of Release, and it has a more intuitive meaning on finite computations.

As is the case with \atlr, the syntax of \atlrd  is defined relative to the following sets:

$\Agt =\{a_1, \ldots, a_n\}$ is a set of $n$ agents,
    $\Res=\{res_1, \ldots,$ $res_r\}$ is a set of $r$ resource types,
    $\Pi$ is a set of propositions, and ${\mathcal B} = 
\nat^{{\Res}^{\Agt}}$ is a set of resource bounds 
(resource allocations to agents). Elements of ${\mathcal B}$
are vectors of length $n$ where each element is a vector of length $r$ (the $k$th
element of the $i$th vector is the allocation of the $k$th resource to the $i$th agent). We will denote by 
${\mathcal B}_A$ (for $A \subseteq \Agt$) the set of possible resource allocations 
to agents in $A$.

Formulas of \atlrd are defined by the following syntax
\[
\phi, \psi ::= p \mid
            \neg \phi \mid
            \phi \lor \psi \mid
            \AO{A^b}{\phi} \mid
            \AU{A^b}{\phi}{\psi} \mid
            \AR{A^b}{\phi}{\psi}
\]
where $p \in \Pi$ is a proposition, $A \subseteq \Agt$, and $b \in {\mathcal B}_A$ 
is a resource bound.
Here, $\AO{A^b}{\phi}$ means that a coalition $A$ can ensure that the next state satisfies $\phi$
under resource bound $b$.
$\AU{A^b}{\phi}{\psi}$ means that $A$ has a strategy to enforce $\psi$ while maintaining the truth
of $\phi$, and the cost of this strategy is at most $b$.
Finally, $\AR{A^b}{\phi}{\psi}$ means that $A$ has a strategy to maintain $\psi$ until and including the time when $\phi$ becomes true,
or to maintain $\psi$ forever if $\phi$ never becomes true, and the cost of this strategy is at most $b$.

The language is interpreted on resource-bounded concurrent game structures. Without loss of generality, we assume that the first resource type is diminishing, \ie is consumed by every action.

\begin{definition}
\label{def:rbcgs}
A resource-bounded concurrent game structure with diminishing resource (RB-CGS$^\spec$) is a tuple
$M = (\Agt$, $\Res$, $S, \Pi, \pi$, $\Act$, $d, c, \delta)$ where:
\begin{itemize}
\item $\Agt$ is a non-empty finite set of $n$ agents, 

\item $\Res$ is a non-empty finite set
  of $r$ resource types, where the first one is the distinguished diminishing resource 

\item $S$ is a non-empty finite set of states;
\item $\Pi$ is a finite set of propositional variables and $\pi : \Pi
  \to \wp(S)$ is a truth assignment which associates each proposition
  in $\Pi$ with a subset of states where it is true;
\item $\Act$ is a non-empty set of actions 

\item $d : S \times \Agt \to \wp(\Act)\setminus \{\emptyset\}$ is a function
  which assigns to each $s \in S$ a non-empty set of actions available
  to each agent $a \in \Agt$. 
  We denote joint actions by all agents in $\Agt$
  available at $s$ by $D(s) = d(s,a_1) \times \cdots \times d(s,a_n)$;
\item $c : S \times \Act \to \integer^r$ is a partial function which
maps a state $s$ and an action $\sigma$ to a vector of
integers, where the integer in position $i$ indicates consumption or production
of resource $r_i$ by the action (negative value for consumption and positive
value for production). We stipulate that the first position in the vector is 
always at most $-1$ (at least one unit of the diminishing resource is consumed by every action).
\item $\delta : S \times \Act^{|\Agt|} \to S$ is a partial function that maps
every $s \in S$ and joint action $\sigma \in D(s)$ to a
state resulting from executing $\sigma$ in $s$.
\end{itemize}
\end{definition}

In what follows, we use the usual point-wise notation for vector comparison and
addition. In particular, $(b_1,\ldots,b_r) \leq (d_1,\ldots,d_r)$ iff $b_i \leq
d_i$ $\forall$ $i \in \{1,\ldots,r\}$,
$(b_1,\ldots,b_r) = (d_1,\ldots,$ $d_r)$ iff $b_i = d_i$ $\forall$ $i \in \{1,\ldots,r\}$, and $(b_1,\ldots,b_r) + (d_1,\ldots,d_r)
= (b_1 + d_1,\ldots,b_r + d_r)$.  We define
$(b_1,\ldots,b_r) < (d_1,\ldots,d_r)$ as $(b_1,\ldots,b_r) \leq (d_1,\ldots,d_r)$ and $(b_1,\ldots,b_r) \not = (d_1,\ldots,d_r)$.
Given a function $f$
returning a vector, we denote by $f_i$ the function that returns the i-th 
component of the vector returned by $f$.

We denote by $\production(s,\sigma)$ the vector obtained by
replacing negative values in $c(s,\sigma)$ by $0$s: it is the vector of
resources produced by action $\sigma$.
We denote by $\consumption(s,\sigma)$ the vector obtained by first replacing
positive values in $c(s,\sigma)$ by $0$s and then replacing negative values by 
their absolute values: $\consumption(s,\sigma) = (|\min(0, c_1(s,\sigma))|$,
\ldots, $|\min(0, c_r(s,\sigma))|)$. It returns the positive costs on each resource of executing
$\sigma$. In particular, $\consumption_1(s,\sigma) \geq 1$.  

We denote the set of all finite non-empty sequences of states (finite computations) 
in a RB-CGS$^{\spec}$ $M$ by $S^+$. We consider only finite computations 
because we
are interested in computations possible under a finite resource bound, and
in the presence of a diminishing resource which is required for any action,
such computations are always finite. For a computation 
$\lambda = s_1\ldots s_k \in S^+$, we use the notation $\lambda[i] = s_i$
for $i \leq k$, $\lambda[i,j] = s_i \ldots s_j$ $\forall$ $1 \leq i \leq j \leq k$, and $|\lambda| = k$ for the length of $\lambda$.

Given a RB-CGS$^\spec$ $M$ and a state $s \in S$, a \emph{joint action by a
coalition} $A \subseteq \Agt$ is a tuple $\sigma = (\sigma_a)_{a \in
  A}$ (where $\sigma_a$ is the action that agent $a$ executes as part of $\sigma$, the $a$th component of $\sigma$) such that $\sigma_a \in d(s,a)$.
For a joint action $\sigma$ by a coalition $A$, we denote by $\consumption(s,\sigma) = (\consumption(s,\sigma_a))_{a \in A}$ the vector of costs of the joint action, similarly for $\production(s,sigma)$.
The set of all joint actions for $A$ at state $s$ is denoted by
$D_A(s)$.

Given a joint action by $\Agt$ $\sigma \in D(s)$, 
$\sigma_A$ (a projection of $\sigma$ on $A$)
denotes the joint action executed by $A$ as part of $\sigma$: $\sigma_A = (\sigma_a)_{a \in A}$. The set of all possible outcomes
of a joint action $\sigma \in D_A(s)$ at state $s$ is:
\[
out(s,\sigma) = \{ s' \in S \mid \exists \sigma' \in D(s): \sigma= \sigma'_A \land
                                           s' = \delta(s,\sigma') \}
                                           \]

A \emph{strategy for a coalition} $A \subseteq \Agt$ in a RB-CGS$^\spec$
 $M$ is a mapping
$F_A : S^+ \to \Act^{|A|}$ such that, for every $\lambda \in S^+$, 
$F_A(\lambda)
\in D_A(\lambda[|\lambda|])$. A computation $\lambda$ is consistent with a strategy $F_A$ iff,
for all $i$, $1 \leq i < |\lambda|$, $\lambda[i+1] \in out(\lambda[i],F_A(\lambda[1,i]))$. We
denote by $out(s,F_A)$ the set of all computations $\lambda$ starting from $s$ that are consistent with $F_A$.

Given a bound $b \in {\mathcal B}_A$, a computation $\lambda \in out(s,F_A)$ is $b$-consistent
with $F_A$ iff, for every $i$, $1 \leq i < |\lambda|$,
\[\consumption(\lambda[i], F_A(\lambda[1,i])) \leq e_A(\lambda[i])\]
where $e_A(\lambda[i])$ is the amount of resources agents in $A$ have in $\lambda[i]$:
$e_A(\lambda[1]) = b$ and 
\begin{align*}
e_A(\lambda[i+1]) =&\  e_A(\lambda[i]) - \consumption(\lambda[i], F_A(\lambda[1,i])) + \\
& \ \production(\lambda[i], F_A(\lambda[1,i])).
\end{align*}
In other words, the amount of resources any of the agents have is never 
negative for any resource type.

A computation $\lambda$ is $b$-maximal for a strategy $F_A$ if
it cannot be extended further while remaining $b$-consistent (the next action prescribed by
$F_A$ would violate $b$-consistency).

The set of all maximal computations starting from state $s$ that are $b$-consistent with $F_A$ 
is denoted by $out(s,F_A,b)$. Note that this set is finite, the maximal length of each computation
is bounded by $b$ (or rather by the 
minimal value for any agent in $A$ of
${b_a}_1$: the bound on the first resource). 

Given a RB-CGS$^\spec$ $M$ and a state $s$ of $M$, the truth of an \atlrd formula $\phi$
with respect to $M$ and $s$ is defined inductively on the structure of $\phi$
as follows:
\begin{itemize}
 \item $M, s \models p$ iff $s \in \pi(p)$;

 \item $M, s \models \neg \phi$ iff $M, s \not\models \phi$;

 \item $M, s \models \phi \lor \psi$ iff $M, s \models \phi$ or $M, s \models
 \psi$;

\item $M, s \models \AO{A^b}{\phi}$ iff $\exists$ strategy $F_A$ such that
      for all $b$-maximal $\lambda \in out(s, F_A, b)$, $|\lambda| \geq 2$ and $M, \lambda[2] \models \phi$;

\item $M, s \models \AU{A^b}{\phi}{\psi}$ iff $\exists$ strategy $F_A$ such
      that for all $b$-maximal $\lambda \in out(s, F_A, b)$, $\exists i$ such that $1 \leq i \leq |\lambda|$, $M, \lambda[i]
      \models \psi$ and $M, \lambda[j] \models \phi$ for all $j \in
      \{1,\ldots,i-1\}$.

\item $M, s \models \AR{A^b}{\phi}{\psi}$ iff $\exists$ strategy $F_A$ such
      that for all $b$-maximal $\lambda \in out(s, F_A,b)$, either $\exists i$ such that $1 \leq i \leq |\lambda|$: 
$M, \lambda[i] \models \phi$ and $M, \lambda[j] \models \psi$ for all $j \in
      \{1,\ldots,i\}$; or, $M, \lambda[j] \models \psi$ for all $j$ such that $1 \leq j \leq |\lambda|$.
\end{itemize}

The most straightforward way of model-checking \atlrd is to adapt the model-checking algorithm
for \atlr \cite{Alechina//:16d} and add a clause for $\AR{A^b}{\phi}{\psi}$. We present this simple algorithm here
because we will use it in modified form in subsequent sections. It is however possible to do \atlrd model-checking more efficiently in the spirit of \cite{DellaMonica//:11a}.


The algorithm is shown in Algorithm \ref{alg:atlrd-label}.
Given a formula, $\phi_0$, we produce a set of subformulas $Sub(\phi_0)$ of $\phi_0$ 
in the usual way. $Sub(\phi_0)$ is ordered in increasing order of complexity. 
We then proceed by cases. For all formulas in $Sub(\phi_0)$ apart from $\AO{A^b}\phi$, $\AU{A^b}{\phi}{\psi}$ and $\AG{A^b}{\phi}$ we essentially run the standard ATL model-checking algorithm \cite{Alur//:02a}. 
Labelling states with $\AO{A^b}\phi$ makes use of a function 
$Pre(A,\rho,b)$ which, given a coalition $A$, a set $\rho \subseteq S$
and a bound $b$, returns a set of states $s$ in which $A$ has a joint action $\sigma_A$ with $\consumption(s,\sigma_A) \leq b$ such that $out(s,\sigma_A) \subseteq \rho$.
Labelling states with $\AU{A^b}{\phi}{\psi}$ and $\AR{A^b}{\phi}{\psi}$ is
more complex, and in the interests of readability we provide separate
functions: \textsc{\untilstrategy} for $\AU{A^b}{\phi}{\psi}$ formulas
is shown in Algorithm \ref{alg:until-strategy}, and
\textsc{\releasestrategy} for $\AR{A^b}{\phi}{\psi}$ formulas is shown in
Algorithm \ref{alg:release-strategy}.

Both algorithms proceed  by depth-first and-or search of $M$. We record information about the state of the search in a search tree of nodes. A \emph{node} is a structure that consists of a state of $M$, the resources available to the agents $A$ in that state (if any), and a finite path (sequence of of nodes and edges) leading to this node from the root node. Edges in the tree correspond to joint actions by all agents. Note that the resources available to the agents in a state $s$ on a path constrain the edges from the corresponding node to be those actions $\sigma_A$ where $\consumption(s,\sigma_A)$ is less than or equal to the available resources.
For each node $n$ in the tree, we have a function $s(n)$ that returns its
state, $p(n)$ that returns the nodes on the path, $\mathit{act}(n)$ that returns the joint action taken to reach $s(n)$ from the preceding state on the path (\ie the edge to $n$), and $e(n)$ that returns the vector of resource availabilities in $s(n)$ for $A$ as a result of following
$p(n)$. The functions $\mathit{act}_{a}(n)$ and $e_{a}(n)$ return the action performed by agent $a \in A$ in $\mathit{act}(n)$ and the resources available to agent $a$ in $e(n)$ respectively.
We use $p(n)[i]$ to denote the $i$-th node in the path $p(n)$, and $p(n)[1, j]$ to denote the prefix of $p(n)$ up to the $j$-th node.
The function $\mathit{node}_0(s,b)$ returns the root node, i.e., a node
$n_0$ such that $s(n_0) = s$, $p(n_0) = [\ ]$, $\mathit{act}(n_{0}) = nil$, 
and $e(n_0) = b$. The function $\mathit{node}(n, \sigma, s')$ returns a node $n'$
where $s(n') = s'$, $p(n') = [p(n) \cdot n]$, $\mathit{act}(n') = \sigma$, 
and for all agents $a \in A$ $e_a (n') = e_a (n) +
\production(s(n),\sigma_a) - \consumption(s(n),\sigma_a)$.

\begin{algorithm}
\caption{Labelling $\phi_0$ }
\label{alg:atlrd-label}
\begin{algorithmic}[1]
\Function{\atlrd-label}{$M, \phi_0$}
\For{$\phi' \in Sub(\phi_0)$}
\Case{$\phi' = p,\ \neg \phi,\ \phi \vee \psi$}\ standard, see \cite{Alur//:02a}
\EndCase
\Case{$\phi' = \AO{A^b}{\phi}$}
\State $[\phi']_M \gets Pre(A, [\phi]_M,b)$
\EndCase
\Case{$\phi' = \AU{A^b}{\phi}{\psi}$}
\State $[\phi']_M \gets \{\ s \mid s \in S \wedge$
\State $\quad \Call{until-strategy}{node_0(s,b), \AU{A^b}{\phi}{\psi} } \}$
\EndCase
\Case{$\phi' = \AR{A^b}{\phi}{\psi}$}
\State $[\phi']_M \gets \{\ s \mid s \in S \wedge$
\State $\quad \Call{release-strategy}{node_0(s,b), \AR{A^b}{\phi}{\psi} } \}$
\EndCase
\EndFor
\State $\mathbf{return\ } [\phi_0]_M$
\EndFunction
\end{algorithmic}
\end{algorithm}

\begin{algorithm}[h]
\caption{Labelling $\AU{A^b}{\phi}{\psi}$ }
\label{alg:until-strategy}

\begin{algorithmic}[1] 
\Function{until-strategy}{$n, \AU{A^b}{\phi}{\psi} $}
\If{$s(n) \in [\psi]_M $}
\State $\mathbf{return}\ \mathit{true}$
\EndIf
\If{$s(n) \not \in [\phi]_M $}
\State $\mathbf{return}\ \mathit{false}$
\EndIf
\State $ActA \gets \{ \sigma \in D_A(s(n)) \mid \consumption(s(n),\sigma) \leq e(n) \}$
\For{$\sigma \in ActA $}
\State $O \gets out(s(n),\sigma)$
\State $\mathit{strat} \gets \mathit{true}$
\For{$s' \in O$}
\State $\mathit{strat} \gets \mathit{strat} \wedge $
\State $\quad \Call{until-strategy}{node(n,\sigma,s'), \AU{A^b}{\phi}{\psi} }$
\EndFor
\If{$ \mathit{strat}$}
\State $\mathbf{return}\ \mathit{true}$
\EndIf
\EndFor
\State $\mathbf{return}\ \mathit{false}$
\EndFunction
\end{algorithmic}
\end{algorithm}

\begin{algorithm}[h]
\caption{Labelling $\AR{A^b}{\phi}{\psi}$ }
\label{alg:release-strategy}

\begin{algorithmic}[1] 
\Function{release-strategy}{$n, \AR{A^b}{\phi}{\psi} $}
\If{$s(n) \in [\psi]_M \cap [\phi]_M$} 
\State $\mathbf{return}\ \mathit{true}$
\EndIf
\If{$s(n) \in [\psi]_M \ \wedge\ $ \NoNumber{\\ $\qquad\ \ \exists \sigma \in D_A(s(n)) : \consumption(s(n),\sigma) \not\leq e(n)$}}
\State $\mathbf{return}\ \mathit{true}$
\EndIf
\If{$s(n) \not \in [\psi]_M $}
\State $\mathbf{return}\ \mathit{false}$
\EndIf
\State $ActA \gets \{ \sigma \in D_A(s(n)) \mid \consumption(s(n),\sigma) \leq e(n) \}$
\For{$\sigma \in ActA $}
\State $O \gets out(s(n),\sigma)$
\State $\mathit{strat} \gets \mathit{true}$
\For{$s' \in O$}
\State $\mathit{strat} \gets \mathit{strat} \wedge $
\State $\quad \Call{release-strategy}{node(n,\sigma,s'), \AR{A^b}{\phi}{\psi} }$
\EndFor
\If{$ \mathit{strat}$}
\State $\mathbf{return}\ \mathit{true}$
\EndIf
\EndFor
\State $\mathbf{return}\ \mathit{false}$
\EndFunction
\end{algorithmic}
\end{algorithm}

When checking whether $\AU{A^b}{\psi_1}{\psi_2}$ or
$\AR{A^b}{\psi_1}{\psi_2}$ is true in a state $s$, we examine paths whose length 
is bounded by the smallest resource bound ${b_a}_1$ on the first resource
in $b$ (since every action costs at least 1 unit of the first resource, any
computation can contain at most ${b_a}_1 $ steps).
An over-approximation of the size of this search tree is $S^{min_{a\in A}({b_a}_1)}$. 
  
\begin{lemma} \label{lem:terminates}
Algorithm \ref{alg:atlrd-label} on input $M$, $\phi$ terminates after at most $O(|\phi|\times |M|^{k})$
steps where $k$ is the maximal value of the first resource bound in $\phi$. 
\end{lemma}

\begin{lemma} \label{lem:correct}
Algorithm \ref{alg:atlrd-label} is correct.
\end{lemma}

\begin{proof}
The Boolean cases of the algorithm are standard.

The algorithm for $\AO{A^b}{\phi}$ returns all states from where there is an action by $A$ that costs
less than $b$ and all outcomes of this action satisfy $\phi$. Essentially in each such state
there is a one-step strategy satisfying $\AO{A^b}{\phi}$. This is all we need because the rest of
actions on this strategy can be arbitrary; the computations that are produced by the strategy do not
need to satisfy any additional constraints apart from being maximal, \ie eventually running out of 
resources (which they are guaranteed to do because of the first resource). 

The algorithm for $\AU{A^b}{\phi}{\psi}$
performs forward and-or search while making sure $\phi$ remains true, until $\psi$ is reached.
It returns true if and only if it finds a strategy where each computation reaches a $\psi$ state 
before $A$ run out of resources to carry on with the strategy, and $\phi$ holds along the computation
up to the point $\psi$ becomes true. Again actions after the $\psi$ state can be arbitrary.

The algorithm for $\AR{A^b}{\phi}{\psi}$ is similar to $\AU{A^b}{\phi}{\psi}$ apart from two points.
One is that the $\psi$ state should also satisfy $\phi$ (the invariant holds not just on the path
to a $\phi$ state but in the $\phi$ state itself). This is ensured by the test at line 2. The second difference
is that there is another way to make $\AR{A^b}{\phi}{\psi}$ true, which is to run out of resources while
maintaining $\psi$. This is the reason for the test at line 4: if the invariant
$\psi$ is true in $s$ and there is an action $\sigma$ in $D_A(s)$ that would cause $A$ to run out of 
resources, we return true because for this computation $\lambda$, the strategy $F_A$ such that
$F_A(\lambda)=\sigma$ ensures that $\lambda$ is a $b$-maximal computation (and it satisfies $\psi$ 
everywhere).
\end{proof}

\begin{theorem}
The model-checking problem for \atlrd is decidable in PSPACE (if resource bounds are written
in unary).
\end{theorem}
\begin{proof}
From Lemmas \ref{lem:terminates} and \ref{lem:correct} we have a model checking algorithm that solves the model checking problem for \atlrd. The complexity which results from the time bound in 
Lemma \ref{lem:terminates} can be improved by observing that the depth first search can be arranged
using a stack and we only need to keep one branch at a time on the stack. The size of the stack is bounded 
by $min_{a \in A}({b_a}_1)$ and hence is polynomial if $b$ is represented in unary.
\end{proof}

\section{\atlrd with imperfect information and perfect recall}
\label{sec:atlrdir}

Agents often have to act under imperfect information, for example, if states are only partially 
observable, an agent may be uncertain whether it is in state $s$ or $s'$. This is represented
in imperfect information models as a binary indistinguishability relation on the set of states for each agent $a$, $\sim_a$: if
$a$ cannot distinguish $s$ from $s'$, we have $s \sim_a s'$. This relation can easily be lifted
to finite sequences of states: if $s_1 \sim_a s'_1$, $s_2 \sim_a s'_2$, then $s_1s_2 \sim_a s'_1 s'_2$.
An essential requirement for strategies under imperfect information is that they are \emph{uniform}:
if agent $a$ is uncertain whether the history so far is $\lambda$ or 
$\lambda'$ ($\lambda \sim_a \lambda'$), then the strategy for $a$ should 
return the same action for both: $F_a(\lambda) = F_a(\lambda')$. Intuitively, 
the agent has no way of choosing different actions in indistinguishable 
situations. A strategy $F_A$ for a group of agents $A$ is uniform if it is 
uniform for every agent in $A$. In what follows, we consider 
\emph{strongly uniform} strategies \cite{Maubert/Pinchinat:14a}, which 
require that a strategy work from all initial states that are indistinguishable by some $a \in A$.

Unfortunately, model-checking for 
ATL under imperfect information with perfect recall uniform strategies, ATL$_{iR}$, is 
undecidable for more than three agents \cite{Dima/Tiplea:11a}. 
It is known that the model checking problem for ATL$_{iR}$
with \emph{bounded} strategies is decidable, 
while for \emph{finite} strategies it is undecidable \cite{Vester:13a}.
Bounded strategies are those that are defined for sequences of states of 
at most some fixed length $k$. In \atlrdir, there is no fixed bound
on the size of strategies, since the size of strategy depends on the formula
and the model. However, we can show that indeed the model checking problem
for \atlrdir with imperfect information and perfect recall strongly
uniform strategies is decidable.


The model checking algorithms are similar to those given for \atlrd in Section \ref{sec:atlr-spec} in that they proceed by and-or depth first search, storing information about the state of the search in a search tree of nodes. However, in this case, the algorithms for Next, Until and Release also take a stack (list) of `open' nodes $B$, a set of `closed' nodes $C$ in addition to an \atlrd formula. $B$ records the current state of the search while $C$ records `successful' branches  (rather than all visited nodes). Uniformity is ensured if action choices are consistent with those taken after $\sim_{a}$ sequences of states on all successful paths explored to date:
($n_1, \ldots, n_k \sim_a n'_1, \ldots, n'_k$ iff $s(n_1), \ldots, s(n_k) \sim_a s(n'_1), \ldots, s(n'_k)$). 
In addition, we assume functions $\hd{u}$, $\tl{u}$ which return the head and tail of a list $u$, and $u\, \circ\, v$ which concatenates the lists $u$ and $v$. (We abuse notation slightly, and treat sets as lists, \eg use $\hd{u}$ where $u$ is a set, to return an arbitrary element of $u$, and use $\circ$ between a set and a list.) 

$M,s \models \AO{A^b}{\phi}$ under strong uniformity requires
 that there exists a uniform
 strategy $F_A$ such that for all $a \in A$, if $s' \sim_a s$, then
for all $b$-maximal $\lambda \in out(s', F_A, b)$: 
$|\lambda| > 1$ and $M, \lambda[2] \models \phi$.
Similarly, in truth conditions for $\AU{A^b}{\phi}{\psi}$ and 
$\AR{A^b}{\phi}{\psi}$ we require the existence of a uniform strategy where 
all $b$-maximal computations starting from states $s'$ indistinguishable 
from $s$ by any $a \in A$ satisfy the Until (respectively, Release) formula.

Weak uniformity only requires the existence of a uniform strategy from $s$.
It is easy to modify the algorithms below to correspond to weak uniformity 
semantics. In fact, the algorithm for $\AO{A^b}{\phi}$ would become
much simpler (identical to that for \atlrd in the previous section).

\begin{algorithm}[h]
\caption{Labelling $\phi_0$}
\label{alg:atlr-label2}
\begin{algorithmic}[1]\small
\Function{\atlrdir-label}{$M, \phi_0$}
\For{$\phi' \in Sub(\phi_0)$}
\Case{$\phi' = p,\ \neg \phi,\ \phi \vee \psi$} \ standard, see \cite{Alur//:02a}
\EndCase
\Case{$\phi' = \AO{A^b}{\phi}$}
\State $[\phi']_M \gets \{\ s \mid s \in S\ \wedge\ $
\StatexIndent[7] $\textsc{\nextstrategy}([node_0(s',b) : s' \sim_{a \in A} s],$ \StatexIndent[9] $ \{\ \}, \AO{A^b}{\phi}) \}$
\EndCase
\Case{$\phi' = \AU{A^b}{\phi}{\psi}$}
\State $[\phi']_M \gets \{\ s \mid s \in S\ \wedge\ $
\StatexIndent[7] $\textsc{\untilstrategy}([node_0(s',b) : s' \sim_{a \in A} s],$ \StatexIndent[9] $\{\ \}, \AU{A^b}{\phi}{\psi} ) \}$
\EndCase
\Case{$\phi' = \AR{A^b}{\phi}{\psi}$}
\State $[\phi']_M \gets \{\ s \mid s \in S\ \wedge\ $
\StatexIndent[7] $\textsc{\releasestrategy}([node_0(s',b) : s' \sim_{a \in A} s],$ \StatexIndent[8] $\ \ \{\ \}, \AR{A^b}{\phi}{\psi} ) \}$
\EndCase
\EndFor
\State $\mathbf{return\ } [\phi_0]_M$
\EndFunction
\end{algorithmic}
\end{algorithm}


\begin{algorithm}[h]
\caption{Labelling $\AO{A^b}{\phi}$} 
\label{alg:x-strategy}
\begin{algorithmic}[1]\small
\Function{\nextstrategy}{$B, C, \AO{A^b}{\phi}$}
\If{$B = [\ ]$}
\State $\mathbf{return}\ \mathit{true}$
\EndIf
\State $n \gets \hd{B}$
\State $Act_A \gets \{\sigma \in D_A(s(n)) \mid \consumption(s(n),\sigma) \leq e(n)\  \wedge\ $  
\StatexIndent[4]  $out(s(n),\sigma) \subseteq [\phi]_M \ \wedge\ \forall a \in A$
\StatexIndent[4] if $\exists n' \in C : p(n) \cdot n \sim_{a} p(n')$ 
\StatexIndent[4] then $\sigma_a = act_a(p(n')[1]) \}$ 
\For{$\sigma \in ActA $}
\If{$\textsc{\nextstrategy}(\tl{B}, C \cup \{ node(n,\sigma,\hd{out(s(n),\sigma)})\},$ \StatexIndent[5] $\ \ \AO{A^b}{\phi})$}
\State $\mathbf{return}\ \mathit{true}$
\EndIf
\EndFor
\State $\mathbf{return}\ \mathit{false}$
\EndFunction
\end{algorithmic}
\end{algorithm}

\begin{algorithm}[h]
\caption{Labelling $\AU{A^b}{\phi}{\psi}$}
\label{alg:until-strategy-imperfect}
\begin{algorithmic}[1] 
\Function{\untilstrategy}{$B, C, \AU{A^b}{\phi}{\psi} $}
\If{$B = [\ ]$}
\State $\mathbf{return}\ \mathit{true}$
\EndIf
\State $n \gets \hd{B}$
\If{$s(n) \in [\psi]_M$}
\State $\mathbf{return}\ \textsc{\untilstrategy}(\tl{B}, C \cup \{n\}, \AU{A^b}{\phi}{\psi})$
\EndIf
\If{$s(n) \not\in [\phi]_M $}
\State $\mathbf{return}\ \mathit{false}$
\EndIf
\State $Act_A \gets \{ \sigma \in D_A(s(n)) \mid \consumption(s(n),\sigma) \leq e(n)\ \wedge\ \forall a \in A$
\StatexIndent[4] if $\exists n' \in C : p(n) \cdot n \sim_{a} p(n')[1, | p(n) \cdot n |]$ 
\StatexIndent[4] then $\sigma_a = act_a(p(n')[| p(n) \cdot n | + 1]) \}$
\For{$\sigma \in Act_A $}
\State $P \gets \{ node(n,\sigma,s') \mid s' \in out(s(n),\sigma) \}$
\If{$\textsc{\untilstrategy}(P \circ \tl{B}, C, \AU{A^b}{\phi}{\psi} )$}
\State $\mathbf{return}\ \mathit{true}$
\EndIf
\EndFor
\State $\mathbf{return}\ \mathit{false}$
\EndFunction
\end{algorithmic}
\end{algorithm}

\begin{algorithm}[h]
\caption{Labelling $\AR{A^b}{\phi}{\psi}$}
\label{alg:release-strategy-imperfect}
\begin{algorithmic}[1] 
\Function{\releasestrategy}{$B, C, \AR{A^b}{\phi}{\psi} $}
\If{$B = [\ ]$}
\State $\mathbf{return}\ \mathit{true}$
\EndIf
\State $n \gets \hd{B}$
\If{$s(n) \in [\psi]_M \cap [\phi]_M$} 
\State $\mathbf{return}\ \textsc{\releasestrategy}(\tl{B}, C \cup \{n\}, \AR{A^b}{\phi}{\psi})$
\EndIf
\If{$s(n) \not \in [\psi]_M $}
\State $\mathbf{return}\ \mathit{false}$
\EndIf
\State $Act_A \gets \{ \sigma \in D_A(s(n)) \mid \forall a \in A$
\StatexIndent[4] if $\exists n' \in C : p(n) \cdot n \sim_{a} p(n')[1, | p(n) \cdot n |]$ 
\StatexIndent[4] then $\sigma_a = act_a(p(n')[| p(n) \cdot n | + 1]) \}$
\For{$\sigma \in Act_A $}
\If{$\consumption(s(n),\sigma) \not\leq e(n)$}
\State $n' \gets node(n,\sigma, s' \in out(s(n),\sigma))$
\If{$\textsc{\releasestrategy}(\tl{B}, C \cup \{n'\}, \AR{A^b}{\phi}{\psi})$}
\State $\mathbf{return}\ \mathit{true}$
\EndIf
\Else
\State $P \gets \{ node(n,\sigma,s') \mid s' \in out(s(n),\sigma) \}$
\If{$\textsc{\releasestrategy}(P \circ \tl{B}, C, \AR{A^b}{\phi}{\psi} )$}
\State $\mathbf{return}\ \mathit{true}$
\EndIf
\EndIf
\EndFor
\State $\mathbf{return}\ \mathit{false}$
\EndFunction
\end{algorithmic}
\end{algorithm}

\begin{lemma}
Algorithm \ref{alg:atlr-label2} terminates in at most $O(|\phi| \times |M|^{k+1})$ steps, where $k$ is the maximal value of the first resource
bound in $\phi$.
\end{lemma}
\begin{proof}

The algorithm for $\AO{A^b}{\phi}$ attempts to find an action which works
(achieves $\phi$) from all states indistinguishable from $s$ by some
agent in $A$. There are at most $|S|$ such states, and at most $|M|$ possible
actions to try. In the worst case (when no action works in all states) we try
every action in each state: $O(|M|^2)$ steps.

As before, the algorithms for $\AU{A^b}{\phi}{\psi}$ and $\AR{A^b}{\phi}{\psi}$ 
are attempting to find a strategy of depth 
$\min_{a \in A} ({b_a}_1)$, but now from all indistinguishable states and 
satisfying additional constraints of uniformity. Considering all 
indistinguishable states adds an additional level (intuitively the root of 
the tree from which all indistinguishable initial states are reachable). 
Satisfying uniformity means having to backtrack to a successful subtree to 
try a different choice of actions even if the previous choice was successful 
(because the same choice does not work in an indistinguishable branch on 
another tree). In the worst case, we will consider all possible actions at each
of $O(b)$ levels of the search tree.  We repeat this for every subformula ($|\phi|$ many times).
\end{proof}

\begin{lemma}
Algorithm \ref{alg:atlr-label2} is correct.
\end{lemma}
\begin{proof}
We consider the cases of $\AO{A^b}{\phi}$, $\AU{A^b}{\phi}{\psi}$ and 
$\AR{A^b}{\phi}{\psi}$.

The algorithm for $\AO{A^b}{\phi}$ places all states which are 
indistinguishable from the current state for one of the agents in $A$ in the 
open list $B$. This ensures that a successful strategy (single action $\sigma$
which is $b$-consistent and achieves $\phi$) found in state $s$ will be placed 
in the closed list $C$, and in states $s' \sim_a s$ (indistinguishable for
the agent $a$) the same action $\sigma_a$ will be attempted as part of
the joint action $\sigma'$ by $A$.  If this does not result in a successful 
strategy in $s'$, the algorithm will backtrack and try another action for
$a$ in $s$. The algorithm returns true if and only if in all indistinguishable
states, an action by $A$ is found which always results in a state satisfying 
$\phi$, is under the resource bound, and its $a$th component is the same
in all $\sim_a$ states. This guarantees that the algorithm found a one step
strategy to satisfy the $\phi$. In order to extend it to an arbitrary uniform
strategy, we can simply select the first action in $D_a(s')$ for all sequences
ending in $s'$ and all $a \in A$. This will ensure that all 
$a$-indistinguishable sequences are assigned the same action.

$\AU{A^b}{\phi}{\psi}$ implements the same idea as above, but with respect to
multi-step strategies. Every time an action is selected on some path $p$,
if $p' \sim_a p$ is in the closed list $C$, then $a$'s action after $p$
is selected to be the same as that selected after $p'$. If this is not 
successful then eventually we will fail back to $p'$ and try a different action
there. If the algorithm returns true, then we are guaranteed that the
strategy contained in $C$ is uniform. We can easily extend the strategy
contained in $C$ to a uniform strategy,  since we do not need to achieve any
objectives after satisfying $\psi$.

$\AR{A^b}{\phi}{\psi}$ is similar to $\AU{A^b}{\phi}{\psi}$, but now we have
an additional complication that actions selected to `run out of resources'
need to be in the closed list since they should also satisfy uniformity. 
This is ensured on lines 11-14 of the algorithm (we add a path ending with an
`expensive' action $\sigma$ and an arbitrary successor $n'$ to the closed 
list). 
\end{proof}
\begin{theorem}
The model-checking problem for \atlrd with imperfect information and perfect recall is decidable in EXPSPACE if the resource bounds are represented in
unary.
\end{theorem}
\begin{proof}
In addition to the space required for the stack, we also need to store the 
closed list $C$. In the worst case, the closed list will contain all possible 
sequences of states of length at most $\min_{a \in A}({b_a}_1)$, which is
$O(|S|^k)$, where k is
the maximal value of the first resource bound in $\phi$.
\end{proof}

\section{\rald} 
\label{sec:rald}

In this section we define a diminishing resource
version of \emph{resource agent logic} (\rald) following \cite{Bulling/Farwer:10a}, with modifications
required for our setting (\eg no infinite endowments). 

The logic is defined over a set of agents $\Agt$, a set of resources types $\Res$, and a set of propositional symbols $\Props$. 

An \emph{endowment (function)} $\eta : \Agt \times \Res \rightarrow \nat$ assigns resources to agents; $\eta_a(r)=\eta(a,r)$ is the amount of resource agent $a$ has of resource type $r$.  $\Enments$ denotes the set of all possible endowments. 

The formulas of RAL$^{\spec}$ are defined by:

\begin{center}$\phi, \psi::=  p \mid \neg \phi \mid \phi \wedge \phi \mid
         \coopdown[B]{A}\Next \phi \mid \coop{A}{}_B^\eta\Next \phi \mid
         \coopdown[B]{A}\phi \NUntil \psi \mid \coop{A}{}_B^\eta\phi \NUntil \psi \mid
 \coopdown[B]{A}\phi \NRel \psi \mid \coop{A}{}_B^\eta\phi \NRel \psi$ 
\end{center} 
where $p \in\Props$ is a proposition, $A, B\subseteq \Agt$ are sets of agents, and  $\eta$ is an endowment. $A$ are called the proponents, and $B$ the (resource-bounded) opponents.

Unlike in \atlrd, in \rald there are two types of cooperation modalities, $\coopdown[B]{A}$ and $\coop{A}{}_B^\eta$. In both types of cooperation modality, the actions performed by agents in $A\cup B$  consume and produce resources (actions by agents in $\Agt \setminus (A \cup B)$ do not change their resource endowment). The meaning
of $\coop{A}{}_B^\eta\varphi$ is otherwise the same as in \atlrd. 
The formula $\coopdown[B]{A}\varphi$ on the other hand requires that
the strategy uses the resources \emph{currently} available to the agents.

The models of RAL$^{\spec}$ are resource-bounded concurrent game structures with diminishing resource (RB-CGS$^\spec$). Strategies are also defined as for \atlrd. However, to evaluate formulas with
a down arrow, such as $\coopdown[B]{A}\Next \varphi$, we need the notion of \emph{resource-extended computations}.
A \emph{resource-extended} computation $\lambda \in (S\times\Enments)^+$ is a non-empty sequence over $S\times\Enments$ such that the restriction to states (the first component), denoted by $\lambda|_S$, is a path in the underlying model.  The projection of $\lambda$ to the second component of each element in the sequence is denoted by $\lambda|_\Enments$.

A \emph{$(\eta,s_A,B)$-computation} is a resource-extended computation $\lambda$ where for all $i=1,\ldots$ with $\lambda[i]:=(s_i,\eta^i)$ there is an action profile $\sigma \in d(\lambda|_\States[i])$ such that:
\begin{enumerate}
\item  $\eta^0 = \enment$ ($\enment$ describes the initial resource distribution);
\item $F_A(\lambda|_\States[1,i])=  \sigma_A$ ($A$ follow their strategy);
\item $\lambda|_S[i+1]=\delta(\lambda|_S[i],\sigma)$ (transition according to $\sigma$);
\item for all $a \in A\cup B$ and $r\in\Res$: $\eta^{i}_a(r) \geq \consumption_r(\lambda|_S[i],\sigma_a)$
 (each agent has enough resources to perform its action);
\item for all $a \in A\cup B$ and $r\in\Res$: $\eta^{i+1}_a(r) = \eta^i_a(r) + \production_r(\lambda|_S[i], \sigma_a) - \consumption_r(\lambda|_S[i], \sigma_a)$ (resources are updated);
\item for all $a \in \Agt \setminus (A \cup B)$ and $r\in\Res$: $\eta^{i+1}_a(r) = \eta^i_a(r)$ (the resources of agents not in $A\cup B$ do not change).
\end{enumerate}
The \emph{$(\eta,B)$-outcome} of a strategy $F_{A}$ in $s$, $\rhooutcome{s}{\eta}{F_{A},B}$, is defined as the set of all  $(\eta,F_{A},B)$-computations starting in $s$.  Truth is defined over a model $\model$, a state $s \in \States$, and an endowment $\eta$.

The semantics is given by the satisfaction relation $\modelsR$  where the cases for propositions, negation and conjunction are standard and omitted:
\begin{description}
\item[$\model,s,\enment \modelsR{\coopdown[B]{A}}\Next \varphi$]
  iff there is a strategy $F_A$ for $A$ such that
  for all $\onepath\in out(s,\enment,F_A,B)$, $|\onepath| > 1$ and $\model,\onepath|_S[2],$ $\onepath|_\Enments[2] \modelsR \varphi$
 
\item[$\model,s,\enment \modelsR{\coop{A}}{}_B^\zeta \Next \varphi$]
  iff there is a strategy $F_A$ for $A$ such that
  for all $\onepath\in out(s,\zeta,F_A,B)$,
$|\onepath| > 1$ and $\model,\onepath|_S[2],$ $\onepath|_\Enments[2] \modelsR \varphi$

\item[$\model,s,\enment \modelsR{\coopdown[B]{A}}\varphi\NUntil\psi$ ]
  iff there is a strategy $F_A$ for $A$ such that
  for all $\onepath\in out(s,\enment,F_A,B)$, 
there exists $i$ with $1 \leq i \leq |\onepath| $
  and
  $\model,\onepath|_S[i],\onepath|_\Enments[i] \modelsR \psi$ and
  for all $j$ with $1 \leq j < i$,
  $\model,\onepath|_S[j],\onepath|_\Enments[j] \modelsR \varphi$

\item[$\model,s,\enment \modelsR{\coop{A}}{}_B^\zeta \varphi\NUntil\psi$]
  iff there is a strategy $F_A$ for $A$ such that
  for all $\onepath\in out(s,\zeta,F_A,B)$,
there exists $i$ with $1 \leq i \leq |\onepath| $
  and
  $\model,\onepath|_S[i],\onepath|_\Enments[i] \modelsR \psi$ and
  for all $j$ with $1 \leq j < i$,
  $\model,\onepath|_S[j],\onepath|_\Enments[j] \modelsR \varphi$

\item[$\model,s,\enment \modelsR{\coopdown[B]{A}}\varphi\NRel \psi$ ]
  iff there is a strategy $F_A$ for $A$ such that
  for all $\onepath\in out(s,\enment,F_A,B)$,
either there exists $i$ with $1 \leq i \leq  |\onepath| $
  and
  $\model,\onepath|_S[i],\onepath|_\Enments[i] \modelsR \psi \wedge \varphi$ and
  for all $j$ with $1 \leq j < i$,
  $\model,\onepath|_S[j],\onepath|_\Enments[j] \modelsR \psi$;
or, for all $j$ with $1 \leq j \leq  |\onepath| $,
  $\model,\onepath|_S[j],\onepath|_\Enments[j] \modelsR \psi$ 

\item[$\model,s,\enment \modelsR{\coop{A}}{}_B^\zeta \varphi\NRel \psi$]
  iff there is a strategy $F_A$ for $A$ such that
  for all $\onepath\in out(s,\zeta,F_A,B)$, either
there exists $i$ with $1 \leq i \leq  |\onepath| $
 and
  $\model,\onepath|_S[i],\onepath|_\Enments[i] \modelsR \psi \wedge \varphi$ and
  for all $j$ with $1 \leq j < i$,
  $\model,\onepath|_S[j],\onepath|_\Enments[j] \modelsR \psi$;
or, for all $j$ with $1 \leq j \leq  |\onepath| $,
  $\model,\onepath|_S[j],\onepath|_\Enments[j] \modelsR \psi$ 

\end{description}


The model checking algorithms for \rald are similar to those given for \atlrd in Section \ref{sec:atlr-spec} in that they proceed by and-or depth first search. However, in this case, the nodes in the search tree also include information about the current proponent and (resource-bounded) opponent coalitions, and the functions that construct nodes are redefined as $\mathit{node}_0(s,b,A,B)$ and $\mathit{node}(n, \sigma, s',A,B)$ where $A$ are the proponents and $B$ are the resource-bounded opponents.

The model checking algorithm for \rald is shown in Algorithm \ref{alg:rald-label}, and takes as input a model
$\model$, a formula $\phi$, and an initial endowment $\eta$, and labels the set of states $[\phi]_\model^\enment$, where $[\phi]_\model^\enment = \{s\ |\ \model , s, \enment \models \phi\}$ is the set of states satisfying $\phi$.
\textsc{\rald-label} simply calls the function \textsc{strategy} to label states with $\phi$. $\mathit{pr}$ and $\mathit{op}$ are functions that return the proponents $A \subseteq \Agt$  and the resource-bounded opponents $B \subseteq \Agt$ respectively if $\phi$ is
of the form $\coop[B]{A}^* \Next \psi$, $\coop[B]{A}^* \psi_1\NUntil\psi_2$, $\coop[B]{A}^* \psi_1 \NRel \psi_2$
where
$*$ is either $\downarrow$ or an endowment,
or $\emptyset$ otherwise. 

\begin{algorithm}
\caption{Labelling $\phi$ }
\label{alg:rald-label}
\begin{algorithmic}[1]\small
\Procedure{\rald-label}{$\model, \phi, \eta $}
\State $[\phi]_{\model}^\enment  \gets \{\ q \mid q \in S\ \wedge$
\StatexIndent[5] $\Call{strategy}{node_0(q,\enment,\mathit{pr}(\phi), \mathit{op}(\phi)), \phi}\}$
\EndProcedure
\end{algorithmic}
\end{algorithm}

The function \textsc{strategy} is shown in Algorithm \ref{alg:rald-strategy} and 
proceeds by depth-first and-or search. We process each coalition modality in turn, starting from the outermost modality. The logical connectives are standard, and simply call \textsc{strategy} on the subformulas. Each temporal operator is handled by a separate function: \textsc{\nextstrategy} for  $\Next \psi$, \textsc{\untilstrategy} for $\phi \NUntil \psi$, and \textsc{\releasestrategy} for $\phi \NRel \psi $.

\begin{algorithm}
\caption{Strategy}
\label{alg:rald-strategy}
\begin{algorithmic}[1]\small
\Function{strategy}{$n, \phi$}
\Case{$\phi = p \in \Pi$}
\State $\mathbf{return}\ s(n) \in \pi(p)$ 
\EndCase
\Case{$\phi = \neg \psi$}
\State $\mathbf{return}\ \neg \Call{strategy}{node_0(s(n),e(n),pr(n),op(n)), \psi}$
\EndCase
\Case{$\phi = \psi_1 \vee \psi_2$}
\State $\mathbf{return}\ \Call{strategy}{node_0(s(n),e(n),pr(n),op(n)), \psi_1}\ \vee$
\StatexIndent[4] $\ \ \Call{strategy}{node_0(s(n),e(n),pr(n),op(n)), \psi_2}$
\EndCase
\Case{$\phi  = \coopdown[B]{A}\Next\psi$}
\State $\mathbf{return}\ \Call{\nextstrategy}{node_0(s(n),e(n),A,B), \phi}$
\EndCase
\Case{$\phi  = \coop[B]{A}^\zeta\Next\psi$}
\State $\mathbf{return}\ \Call{\nextstrategy}{node_0(s(n),\zeta,A,B), \phi}$
\EndCase
\Case{$\phi  = \coopdown[B]{A}\psi_1\NUntil\psi_2$}
\State $\mathbf{return}\ \Call{\untilstrategy}{node_0(s(n),e(n),A,B), \phi}$
\EndCase
\Case{$\phi  = \coop[B]{A}^\zeta\psi_1\NUntil\psi_2$}
\State $\mathbf{return}\ \Call{\untilstrategy}{node_0(s(n),\zeta,A,B), \phi}$
\EndCase
\Case{$\phi  = \coopdown[B]{A} \psi_1 \NRel \psi_2$}
\State $\mathbf{return}\ \Call{\releasestrategy}{node_0(s(n),e(n),A,B), \phi}$
\EndCase
\Case{$\phi  = \coop[B]{A}^\zeta \psi_1 \NRel \psi_2$}
\State $\mathbf{return}\ \Call{\releasestrategy}{node_0(s(n),\zeta,A,B), \phi}$
\EndCase
\EndFunction
\end{algorithmic}
\end{algorithm}

The function \textsc{\nextstrategy} for formulas of types $\coop[B]{A}^\downarrow \Next \phi$ and $\coop[B]{A}^\zeta \Next \phi$ is shown in Algorithm \ref{alg:rald-next}
and is straightforward. We simply check if there is an action of $A$ that is possible given the current endowment (lines 2--4), and where in all outcome states $A$ has a strategy to enforce $\phi$ (lines 6--10). Note that the recursive call (line 8) is to \textsc{strategy}, to correctly determine the endowments for the new search in both the case where $\phi$ specifies a fresh endowment or the resources currently available to the agents (i.e., down arrow). 

\begin{algorithm}
\caption{Next (both types of modalities) }
\label{alg:rald-next}
\begin{algorithmic}[1]\small
\Function{\nextstrategy}{$n, \coop[B]{A}^* \Next \phi$}
\State $Act_{A} \gets \{ \sigma' \in D_{A} (s(n)) \mid \consumption(\sigma') \leq e_{A}(n) \}$
\For{$\sigma' \in Act_{A}$}
\State $Act_{Agt}  \gets \{\sigma \in D(s(n))\mid \sigma_{A} = \sigma' \wedge\ $
\StatexIndent[9] $\ \ \consumption(\sigma_B) \leq e_{B}(n) \}$
\State $\mathit{strat} \gets \mathit{true}$
\For{$\sigma \in Act_{Agt}$}
\State $s' \gets \delta(s(n),\sigma)$
\State $\mathit{strat} \gets \mathit{strat}\ \wedge\ \Call{strategy}{node(n, \sigma, s', A, B), \phi }$ 
\EndFor
\If{$ \mathit{strat}$}
\State $\mathbf{return}\ \mathit{true}$
\EndIf
\EndFor
\State $\mathbf{return}\ \mathit{false}$
\EndFunction
\end{algorithmic}
\end{algorithm}

The function \textsc{\untilstrategy} for formulas of types $\coop[B]{A}^\downarrow \phi \NUntil\psi$ and $\coop[B]{A}^\zeta \phi\NUntil\psi$ is shown in Algorithm \ref{alg:rald-until}. If $A$ have a strategy to enforce $\psi$, we return true (lines 2--3). We then check if it is possible to enforce $\phi$ in $n$, and terminate the search with false if it is not (lines 4--5). Otherwise the search continues. Each action available at $s(n)$ is
considered in turn (lines 6--14). For each action $\sigma' \in Act_A$, we check whether a recursive call of the algorithm returns true in all outcome states $s'$ of $\sigma'$  (i.e., $\sigma'$ is part of a successful strategy). If such a $\sigma'$ is found, the algorithm returns true. Otherwise the algorithm returns false. 

The function \textsc{\releasestrategy} for formulas of types $\coop[B]{A}^\downarrow \phi \NRel\psi$ and $\coop[B]{A}^\zeta \phi\NRel\psi$ is similar (see Algorithm \ref{alg:rald-release}).

\begin{algorithm}
\caption{Until (both types of modalities)}
\label{alg:rald-until}
\begin{algorithmic}[1]\small
\Function{\untilstrategy}{$n, \coop[B]{A}^* \phi \NUntil \psi $}
\If{$\Call{strategy}{n, \psi}$}
\State $\mathbf{return}\ \mathit{true}$
\EndIf
\If{$\neg\, \Call{strategy}{n, \phi}$}
\State $\mathbf{return}\ \mathit{false}$
\EndIf
\State $Act_{A} \gets \{ \sigma' \in D_{A}(s(n)) \mid \consumption(\sigma') \leq  e_A(n) \}$
\For{$\sigma' \in Act_{A} $}
\State $Act_{Agt}  \gets \{\sigma \in D(s(n))\mid \sigma_A = \sigma' \wedge\ $
\StatexIndent[9.5] $ \consumption(\sigma_B) \leq e_{B}(n) \}$
\State $\mathit{strat} \gets \mathit{true}$
\For{$\sigma \in Act_{Agt}$}
\State $s' \gets \delta(s(n),\sigma)$
\State $\mathit{strat} \gets \mathit{strat}\ \wedge \ $
\StatexIndent[5.5] $\Call{\untilstrategy}{node(n,\sigma,s', A,B), \coop[B]{A}^* \phi \NUntil \psi}$
\EndFor
\If{$ \mathit{strat}$}
\State $\mathbf{return}\ \mathit{true}$
\EndIf
\EndFor
\State $\mathbf{return}\ \mathit{false}$
\EndFunction
\end{algorithmic}
\end{algorithm}

\begin{algorithm}
\caption{Release (both types of modalities)}
\label{alg:rald-release}
\begin{algorithmic}[1]\small
\Function{\releasestrategy}{$n, \coop[B]{A}^* \phi \NRel \psi $}
\If{$\neg \Call{strategy}{n, \psi}$}
\State $\mathbf{return}\ \mathit{false}$
\EndIf
\If{$\Call{strategy}{n, \phi}$}
\State $\mathbf{return}\ \mathit{true}$
\EndIf

\If{$\exists\, \sigma \in D_{A} \text{\ s.t.\ } \consumption(s(n),\sigma) \not\leq e_{A}(n))$}
\State $\mathbf{return}\ \mathit{true}$
\EndIf

\State $Act_{A} \gets \{ \sigma' \in D_{A}(s(n)) \mid \consumption(\sigma') \leq e_A(n) \}$
\For{$\sigma' \in Act_{A} $}
\State $Act_{Agt}  \gets \{\sigma \in D(s(n))\mid \sigma_A = \sigma' \wedge\ $
\StatexIndent[9.5] $ \consumption(\sigma_B) \leq e_{B}(n) \}$
\State $\mathit{strat} \gets \mathit{true}$
\For{$\sigma \in Act_{Agt}$}
\State $s' \gets \delta(s(n),\sigma)$
\State $\mathit{strat} \gets \mathit{strat}\ \wedge \ $
\StatexIndent[5.5] $\Call{\releasestrategy}{node(n,\sigma,s', A,B), \coop[B]{A}^* \phi \NRel \psi}$
\EndFor
\If{$ \mathit{strat}$}
\State $\mathbf{return}\ \mathit{true}$
\EndIf
\EndFor
\State $\mathbf{return}\ \mathit{false}$
\EndFunction
\end{algorithmic}
\end{algorithm}

\begin{lemma}
Algorithm \ref{alg:rald-strategy} terminates in $O(|M|^{|\phi|})$ steps, where the bounds in $\phi$ are written in unary.
\end{lemma}
\begin{proof}
The only difference between the RAL$^\spec$ algorithms and the algorithms in
section \ref{sec:atlr-spec} is the fact that in the case of RAL$^\spec$
we cannot label states with subformulas. For example, we cannot find states 
satisfying $\coopdown[B]{A}\phi \NUntil \psi$ because we do not know which 
endowment
the $\downarrow$ refers to. When verifying a formula with non-propositional
subformulas, for example $\coopdown[B]{A}\phi \NUntil \psi$ again, where $\phi$
and $\psi$ are not propositional, we have to make recursive calls to check
whether the current state satisfies $\phi$ or $\psi$ \emph{with the current
endowment}. Hence the checks for $\Call{strategy}{n, \phi}$ instead of
checking whether $s(n) \in [\phi]_M$. However the recursive calls are always
to formulas of lower complexity, and it is easy to show that in the
propositional case they do terminate, and that under the inductive assumption
if lower complexity calls terminate, then the calls to 
$\coop[B]{A}^* \Next \phi$, $\coop[B]{A}^* \phi \NUntil \psi$ and 
$\coop[B]{A}^* \phi \NRel \psi$ terminate.

The algorithm again performs depth first and-or search, but now up to the depth determined by the nestings of modalities in $\phi$: we need to take the sum
of the minimal bounds for the first resource occurring in the endowment of 
some resource bounded agent in nested formulas to find the maximal depth of 
the tree. We can ignore $\downarrow$ endowments because they will use the 
amount of the first resource remaining from the outer modalities.
\end{proof}

\begin{lemma}
Algorithm \ref{alg:rald-strategy} is correct.
\end{lemma}
\begin{proof}
Assuming that calls to $\Call{strategy}{n,\phi}$ terminate and have the same
effect as checking whether $s(n) \in [\phi]_M$, the algorithms are the same
as for \atlrd. The only small difference is that we remember the current
endowment and pass it to the $\downarrow$ modalities as if it was an explicit
bound $b$ in \atlrd.
\end{proof}

\begin{theorem}
The model-checking problem for   RAL$^\spec$ is decidable in PSPACE (if resource bounds are written in unary).
\end{theorem}
\begin{proof}
From the two lemmas above it follows that Algorithm \ref{alg:rald-strategy} is
a terminating and correct model-checking algorithm for RAL$^\spec$. The space
it is using on the stack is polynomial in the size of the formula (it is the
sum of nested resource bounds on the first resource for the minimally endowed 
agents). After at most $O(k)$ steps, where $k$ is the maximal value of the first resource bound in $\phi$, the endowment becomes negative for one
of the agents, and the algorithm terminates.
\end{proof}

\section{Conclusion}

In this paper we studied resource logics over models with a diminishing
resource.  We gave new and simple 
model-checking algorithms for the versions of \atlr, \atlrir and RAL
with a diminishing resource. We believe that settings where one of
the resources is always consumed are quite common, and our results
may therefore be of practical interest. It was known that the model checking
problem for \atlr is decidable, but our complexity result for \atlrd is new.
Decidability of the model checking problem for RAL follows from a more general
result on bounded models from \cite{Bulling/Farwer:10a}, but no
model checking algorithm was given there. 
The model checking algorithm for RAL$^{\spec}$ is different from the
algorithm for the decidable fragment of RAL presented in \cite{Alechina//:17b} because
it works for the full RAL rather than just for the positive fragment of proponent-restricted RAL in \cite{Alechina//:17b}.

\bibliographystyle{ACM-Reference-Format}  
\bibliography{references}  

\end{document}